\documentclass{article}
\usepackage{microtype}
\usepackage{wrapfig}
\usepackage{graphicx}
\usepackage{enumerate}
\usepackage{amssymb,amsthm}
\usepackage{subfig,wrapfig}
\usepackage{amsfonts}
\usepackage{color}
\usepackage{xspace}
\newcommand{\cancel}[1]{}

\newcommand{\myleft}{\mathrm{left}}
\newcommand{\myright}{\mathrm{right}}

\newcommand{\mycomment}[2]{\footnote{\textcolor{blue}{#1 #2}}}
\newif\ifComments
\Commentstrue
\ifComments
  \newcommand{\hung}[1]{\mycomment{Hung:}{#1}}
  \newcommand{\marcel}[1]{\mycomment{Marcel:}{#1}}
  \newcommand{\matias}[1]{\mycomment{Matias:}{#1}}
\else
  \newcommand{\hung}[1]{}
  \newcommand{\marcel}[1]{}
  \newcommand{\matias}[1]{}
\fi

\newcommand{\etal}{\emph{et al.}\xspace}
\newcommand{\tsat}{\textsc{3SAT}}
\newcommand{\ptsat}{\textsc{planar 3SAT}}
\newcommand{\calC}{{\cal C}}

\newcommand{\calF}{{\cal F}}

\theoremstyle{plain}
\newtheorem{theorem}{Theorem}
\newtheorem{lemma}[theorem]{Lemma}

\newtheorem{corollary}[theorem]{Corollary}

\graphicspath{{figure/}}

\title{Line Segment Covering of Cells in Arrangements\thanks{A preliminary version of this paper appeared in the proceedings of the COCOA 2015 conference~\cite{cocoa}}}
\author{Matias Korman \and Sheung-Hung Poon \and Marcel Roeloffzen}
\begin{document}

\maketitle

\begin{abstract}
Given a collection $L$ of line segments, we consider its arrangement and study the problem of covering all cells with line segments of $L$. That is, we want to find a minimum-size set $L'$ of line segments such that every cell in the arrangement has a line from $L'$ defining its boundary. We show that the problem is NP-hard, even when all segments are axis-aligned. In fact, the problem is still NP-hard when we only need to cover rectangular cells of the arrangement. For the latter problem we also show that it is fixed parameter tractable with respect to the size of the optimal solution. Finally we provide a linear time algorithm for the case where cells of the arrangement are created by recursively subdividing a rectangle using horizontal and vertical cutting segments.
\end{abstract}

%%%%%%%%%%%%%%%%%%%%%%%%%%%%%%%%%%%%%%%%%%%%%%%%%%%%%
\section{Introduction}
\label{sec:Introduction}
%%%%%%%%%%%%%%%%%%%%%%%%%%%%%%%%%%%%%%%%%%%%%%%%%%%%%

Set cover~\cite{clrs} is one of the most fundamental problems of computer science. This problem is usually formulated in terms of hypergraphs: the input of the problem is a hypergraph $\mathcal{H}=(X,\calF)$ where $\calF \subseteq 2^X$ is a collection of subsets of $X$, and we aim for a subset $\calF' \subseteq \calF$ of smallest cardinality that {\em covers} $X$ (i.e., $\cup_{F \in \calF'} F = X$). This problem is known to be NP-hard and even hard to approximate~\cite{clrs,f1998}.

Given its importance, it is not surprising that this problem has been studied extensively. In most cases, the set $\calF$ is given implicitly (this is specially true when considering geometric variants of the problem). For example, in the well-known $k$-center problem~\cite{inapp-center} we want to cover a set $S$ of $n$ points with unit disks. In the hypergraph definition, this is equivalent to $X=S$ and $\calF$ is the collection of subsets of $S$ that can be covered with a single unit disk.

Sometimes the relationship between $X$ and $\calF$ is much more involved. For example, in the {\em discrete center problem}, we only consider the disks whose center is a point of $S$. Akin to the discrete variant of the $k$-center problem, in this paper we study a geometric setting where the elements $X$ and sets $\calF$ are defined by the same geometric primitives. Specifically, we study the problem of covering the cells of an arrangement of line segments $L$ with segments of $L$. Given a set $L$ of line segments in the plane a \emph{cell} in the arrangement of $L$ is defined as a maximally connected region that is not intersected by any segments of $L$. Essentially the cells are the `empty'---not intersected by segments of $L$---regions in the arrangement defined by $L$. Now let $C$ denote the set of all cells in the arrangement of $L$. We say that a cell $c \in C$ is {\em covered} by a line segment $\ell \in L$ if and only if $\ell$ is part of the boundary of $c$. Similarly $c$ is covered by a set $L'$ of line segments if and only if there is a segment $\ell \in L'$ that covers $c$. The goal is then to find a minimum-size set $L' \subset L$ that covers all cells of $C$. We call this the \emph{line-segment covering} problem.

The problem can also be viewed as a guarding problem. In the traditional art gallery problem, the goal is to place guards so that the guards together see the whole gallery (often a simple polygon). Many variants of this have been studied. Bose~\etal~\cite{bcchklt} study guarding and coloring problems between lines. They provide results for several types of guards and objects to guard, such as guarding the cells of the arrangement with the lines, or guarding the lines by selecting cells. Their results however do not extend to line segments as they use properties of the lines that do not hold for line segments. To the best of our knowledge covering cells in an arrangement of line segments with the segments has not been studied before.

We study three different variants of this problem. First, in Section~\ref{sec:npc-rectilinear} we show that the line-segment covering problem is NP-hard, even when all segments of $L$ are axis-aligned. In Section~\ref{sec:fpt-algo} we consider a slightly different variant, where we are required to cover only rectangular cells, those defined by four line segments. For this variant we show that the NP-hardness reduction still works. However, we show that this variant is fixed parameter tractable with respect to the size of the optimal solution. In Section~\ref{sec:rect-bsp}, inspired by subdivisions induced by KD-trees, we study a variant where the line segments define a type of rectangular subdivision. That is, an axis aligned rectangle that is recursively subdivided with horizontal or vertical line segments, similar to the subdivision defined by a KD-tree~\cite{bcko}. For this case we show that an optimal cover can be computed in linear time, assuming that the partitioning is given as a tree-structure defined by the splitting lines.

%%%%%%%%%%%%%%%%%%%%%%%%%%%%%%%%%%%%%%%%%%%%%%%%%%%%
\section{NP-hardness for Rectilinear Line Segments}
\label{sec:npc-rectilinear}
%%%%%%%%%%%%%%%%%%%%%%%%%%%%%%%%%%%%%%%%%%%%%%%%%%%%%%%

\begin{figure}[tb]
  \centering
  \includegraphics{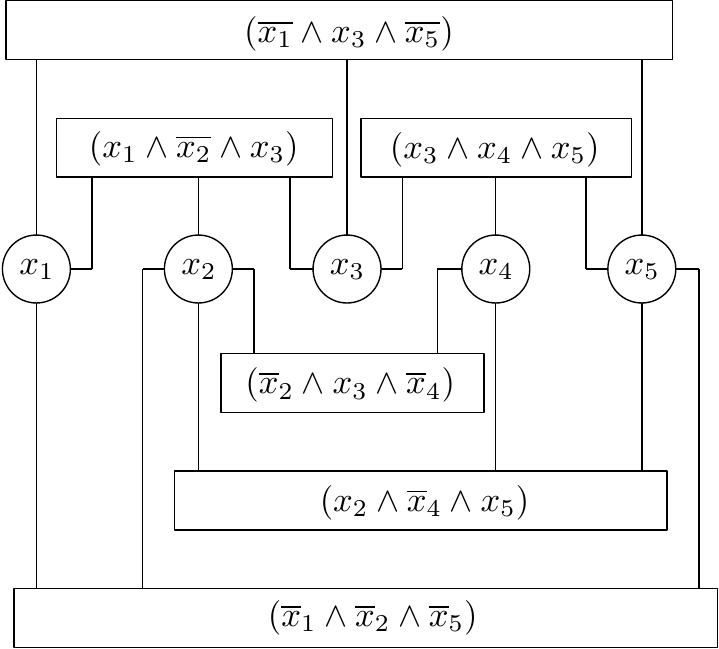}
  \caption{$\ptsat$ problem instance along with a planar embedding.}
  \label{fig:planar3sat}
\end{figure}

In this section we show that the line-segment covering problem is NP-hard, even if the input consists of only horizontal and vertical line segments. We reduce the problem from $\ptsat$~\cite{planar-3SAT}. An input instance for the $\tsat$ problem is a set $\{ x_1, x_2, \dots, x_n \}$ of $n$ variables and a Boolean expression in conjunctive normal form. That is, the expression is a conjunction of clauses $\Phi= c_1 \wedge \ldots \wedge c_m$ such that each clause $c_i$ is a disjunction of three literals (a variable or negation of a variable). The problem is then to decide if there is a truth assignment for the variables so that $\Phi$ is true. In $\ptsat$ we impose further restrictions by looking at the representation of $\Phi$ as a bipartite graph with variables and clauses as vertices. A variable-node $v$ is connected to a clause-node $c$ if and only if $v$ occurs in $c$. In $\ptsat$ we assume that this graph is planar. Specifically we assume that a planar embedding is given that places all variable-nodes on a horizontal line and all clause-nodes above or below this line (see Fig.~\ref{fig:planar3sat}). We also assume that no variable appears more than once in any clause (that is, the above described bipartite graph is a proper graph and not a multigraph).

It is well-known that the $\ptsat$ problem is NP-hard~\cite{planar-3SAT}. Also note that it is easy to see that the line-segment covering problem is in NP. Indeed, given a possible covering, we can construct the arrangement of line segments and verify in polynomial time that indeed all cells are covered. In the remainder of this section we provide a polynomial time reduction and prove its correctness.

\subsection{Reduction}
\begin{figure}[tb]
  \centering
  \includegraphics{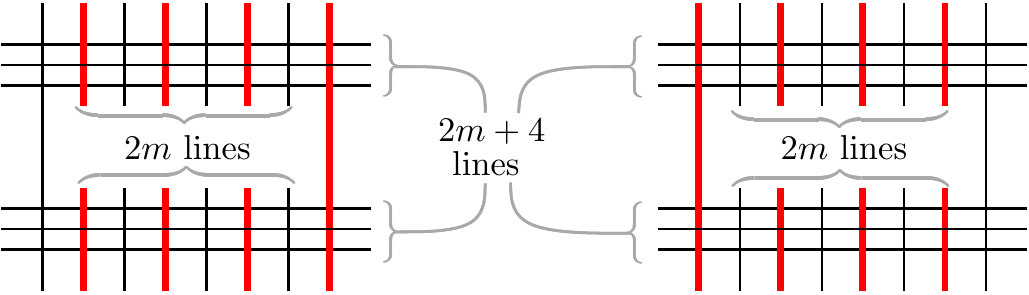}
  \caption{Gadget for a variable with a true (left) and false (right) assignment. Red (thick) edges show the two possible covers with $m+1$ segments.}
  \label{fig:variable}
\end{figure}
\begin{figure*}[tb]
  \centering
  \includegraphics{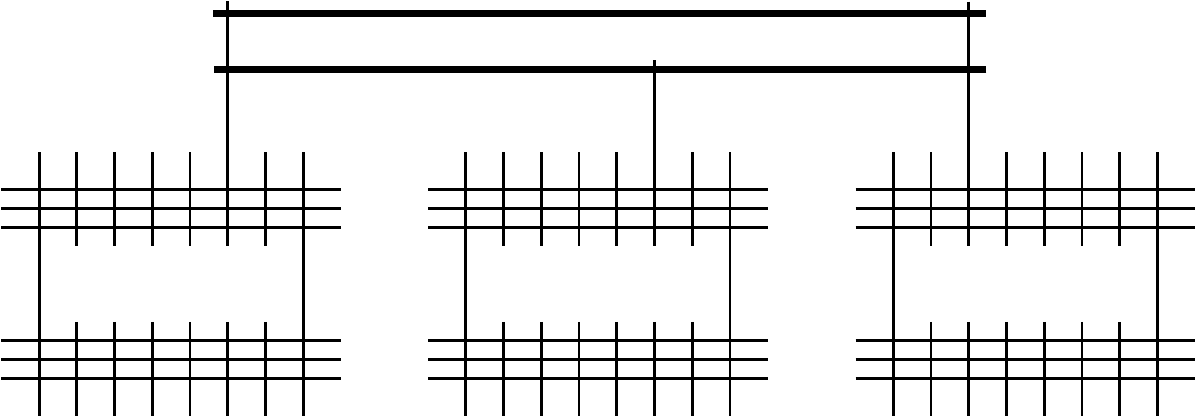}
  \caption{A clause gadget used in showing NP-hardness in Section~\ref{sec:npc-rectilinear}, parts are variable gadgets that connect to the edges of a clause gadget (marked with thicker line segments).}
  \label{fig:clause}
\end{figure*}

For each variable in $\Phi$ we create a gadget consisting of $4m + 8$ horizontal segments and $4m+2$ vertical segments. The leftmost and rightmost vertical line stab all horizontal lines of the gadget, whereas $2m$ of the other lines stab the top $2m+4$ horizontal lines and $2m$ stab the lower $2m+4$ horizontal lines as illustrated in Fig.~\ref{fig:variable}. As we describe later, some of the vertical line segments may be further extended (above or below) to connect to the clause gadgets, but the horizontal segments will not cross any segments of other gadgets.

Intuitively speaking, we will show that any covering of the variable gadgets must choose one every other vertical segment, including either the right or leftmost segment. The choice of using either the rightmost or leftmost vertical segment is equivalent to assigning the variable to be true or false. The clause gadget will create additional cells that will be covered {\em for free} (without selecting additional line segments) provided that at least one variable satisfies the clause.

Let $s^{(i)}_1, \ldots, s^{(i)}_{2m}$ be the vertical segments created in the gadget for variable $x_i$ (numbered from left to right). We would like to sort the clauses $c_1, \ldots, c_j$ in which $x_i$ occurs in the order in which they appear on the embedding of the $\ptsat$ instance. However, this is not well-defined (since it is not always clear when a segment goes before another), so we proceed as follows: let $c_1, \ldots, c_{j'}$ be the clauses that contain variable $x_i$ and are embedded above the line containing all variable nodes, sorted in clockwise order of their connections to $x_i$. Similarly, let $c_{j'+1}, \ldots, c_j$ be the clauses that are embedded below the line (this time in counter-clockwise order). We define the ordering of the clauses around $x_i$ as the concatenation of both orderings. Since we have $2m$ vertical line segments for each clause and $m$ clauses we ensure that any vertical segment of a variable gadget is extended only towards a single clause, and any clause is associated to exactly three segments.

We now detail the gadget associated to clause $c=\ell_i \wedge \ell_j \wedge \ell_k$ (for $i<j<k$), where $\ell_i$ is a literal of variable $x_i$ (similarly, $\ell_j$ and $\ell_k$ are literals of variables $x_j$ and $x_k$, respectively). First, we extend the three segments associated to clause $c$ (above or below depending on where $c$ is placed in the embedding). We extend the segments associated to variables $x_i$ and $x_k$ slightly further than the segment of $x_j$. We complete our transformation by adding two horizontal segments that create a rectangle with the three extended segments, see Fig.~\ref{fig:clause}.

This concludes the construction of a line-segment-covering instance $L$ from a $\ptsat$ input $\Phi$. Next we show that there is a satisfying assignment for $\Phi$ if and only if there is a subset $L'$ of $L$ of size at most $n(2m+1)$ that covers all cells in the arrangement.

\subsection{Correctness}
In the following theorem, we show that the line segment covering problem is NP-complete even when the input line segments are rectilinear.

\begin{lemma}\label{lem:npc-equiv}
A $\ptsat$ expression $\Phi$ is satisfiable if and only if there is a cover of size at most $n(2m+1)$ for its corresponding line-segment-covering instance $L$.
\end{lemma}
\begin{proof}
First we prove that given a satisfying assignment for $\Phi$ we can cover $L$ with $n(2m+1)$ segments. If a variable is true, then we select the rightmost vertical segment, and for each set of $2m$ segments that only intersect the top or bottom we select the odd ones counting from the leftmost segment starting at one, see also Fig.~\ref{fig:variable}. If the variable is false we select the leftmost longer segment and the even ones from the sets of shorter segments.

Next we show that all cells are covered. We consider three types of cells: cells in the interior of the grids created of the variable gadgets are called {\em variable cells}, the single rectangular cell associated to a clause gadget is called {\em clause cell}; any other cell (included the unbounded one) that is created with our construction is simply called an {\em other cell}.

Since we have selected one every other segment, clearly all variable cells are covered. The fact that the variable assignment satisfies all clauses implies that at least one of the three vertical segments defining a clause cell has been selected, so the clause cell is covered. Hence, all clause cells are also covered. Finally, for the remaining cells it suffices to see that each such cell always has two consecutive vertical segments of a variable gadget in its boundary. Indeed, Such cells are only created when connecting clauses and variables and in particular, their left and right boundaries are created by those extensions. Thus, when walking along the boundary of any such cell, we will find the next vertical segment of the variable gadget (or the predecessor in case the segment was the last one). One of the segments must have been selected, so also these cells are covered.

The reverse statement is similar. Assume that we have a cover $L'$ for $L$ of size $n(2m+1)$. First observe that each variable gadget needs at least $2m+1$ selected line segments to cover its interior cells. To achieve this we must select either the left or rightmost segment, after which there is a unique cover for the remaining cells that uses only $2m$ segments. Covering the cells within the variable gadgets with fewer than $2m+1$ segments is not possible, so any cover consist of exactly $2m+1$ segments per variable gadget. Furthermore, none of these segments can be reused between different variable gadgets. Thus, we conclude that each clause gadget must be covered by the lines selected from the variable gadgets. We create a variable assignment for each variable as before, depending if the leftmost or rightmost segments has been selected.

Since $L'$ covers all cells, it must also cover the clause cells, which implies that at least one of the three vertical segments has been selected. Equivalently, this implies that in each clause the choice of assignment of the variables makes at least one of its literals true and the formula $\Phi$ is satisfiable.
\end{proof}

Since the reduction is easily computed in polynomial time we conclude the following result.

\begin{theorem}
Given a set $L$ of axis-aligned line segments, it is NP-hard to find a minimum-size set $L' \subseteq L$ so that for each cell of the arrangement at least one of its defining segments is in $L'$.
\end{theorem}

%%%%%%%%%%%%%%%%%%%%%%%%%%%%%%%%%%%%%%%%%%%%%%%%%%%%%%%
\section{Covering only rectangular cells}
\label{sec:fpt-algo}
%%%%%%%%%%%%%%%%%%%%%%%%%%%%%%%%%%%%%%%%%%%%%%%%%%%%%%%
From the above reduction we can see that the main difficulty of the problem lies in covering the rectangular cells. Thus, in this section we turn our attention to a variant of the problem in which segments are axis-aligned and we are not required to cover all cells, but only those that are rectangles. That is, cells whose boundary is formed by exactly four line segments. First we briefly argue that this variant is also NP-hard by adapting the NP-hardness proof in the previous section. Then we show that the problem is fixed parameter tractable (FPT) with respect to $k$, the number of segments in the optimal solution.

\subsection{NP-hardness}
The hardness almost follows from the construction of Section~\ref{sec:npc-rectilinear}. Indeed, clause cells are the only critical part of the reduction that need to be modified. Instead, we create the clause gadget with 6 segments as shown in Figure~\ref{fig:clause-gadget-updated}. This modified gadget contains three rectangular cells. Note that incoming segments from variables can cover at most two of these cells, but at least one segment must be added so as to cover the intermediate rectangular cell. This additional edge can cover two cells, either the left and middle cells, or the middle and right cells. Thus, it follows that we can find a covering of all rectangular cells of this modified instance with $n(2m+1)+m$ segments if and only if the associated $\ptsat$ instance is satisfiable, and thus this variation is also NP-hard.

\begin{figure}
\centering
\includegraphics{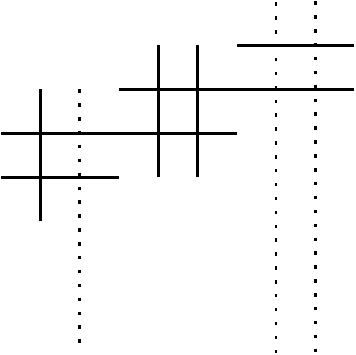}
\caption{Modified clause gadget. The three dashed vertical segments connect to the corresponding variable gadgets. This new gadget creates three rectangular cells that can be guarded with one additional segment if and only if the variable assignment satisfies the clause.}
\label{fig:clause-gadget-updated}
\end{figure}

\begin{theorem}
Given a set $L$ of axis-aligned line segments, it is NP-hard to find a minimum-size set $L' \subseteq L$ so that for each rectangular cell of the arrangement at least one of its defining segments is in $L'$.
\end{theorem}

\subsection{FPT on the size of the optimal solution}
Next we show that the problem is fixed parameter tractable (FPT) with respect to $k$, the size of the optimal solution. Our aim is to compute a kernel of small size (or conclude that there is no solution of size at most~$k$).

Since we want to cover only rectangular cells we can represent each cell by an {\em associated subset} (or subset for short) $C = \{ \ell_1, \ell_2, \ell_3, \ell_4\}$ with the four bounding line segments as its elements. This reduces the line segment covering problem to a hitting set problem for a collection $\calC$ of subsets of size four. Our approach is to reduce the number of subsets to consider; first to a set $\calC_1$ where for any two line segments there are at most $2k$ subsets that contain both these line segments; then to a set $\calC_2$ where for any single line segment there are at most $2k^2$ subsets containing it.
First we prove the following lemma.

\begin{lemma}\label{lem:common3}
Let $\ell, \ell'$ and $\ell''$ be any three line segments in $L$.
There are at most two subsets in $\calC$ containing all three line segments,
$\ell, \ell'$ and $\ell''$.
\end{lemma}
\begin{proof}
Since the arrangement is rectilinear, two lines, say $\ell$ and $\ell'$ are parallel and the other is orthogonal to these. This means that any cell having all three line segments $\ell, \ell'$ and $\ell''$ on its boundary must span the strip between $\ell$ and $\ell'$. However at most two such cells can also be adjacent to the third line segments $\ell''$.
\end{proof}

We start reducing our problem instance by looking at pairs of line segments. Specifically we count for every pair of line segments how many subsets contain both. Then for any pair $\ell, \ell'$ shared in more than $2k$ subsets we add the subset $\{\ell, \ell'\}$ and remove all subsets containing both $\ell$ and $\ell'$. Let $\calC_1$ denote this reduced subset.

\begin{lemma}\label{lem:common2}
A set $L' \subset L$ of line segments, with $|L'| \leq k$ is a minimum-size cover of $\calC_1$ if and only if it is a minimum-size cover of $\calC$.
\end{lemma}
\begin{proof}
Clearly the claim holds if $\calC=\calC_1$. Thus, from now on we assume that $\calC^+ = \calC_1 \backslash \calC$ and $\calC^- = \calC \backslash \calC_1$ are two nonempty sets that contain all elements that were added and removed from $\calC$, respectively.

Now assume that $L'$ with $|L'| \leq k$ is a minimum size cover for $\calC_1$. Observe that all subsets of $\calC^-$ must also be covered, since for every subset $C \in \calC^-$ there is a subset $\{\ell, \ell'\}$ such that both $\ell$ and $\ell'$ occur in $C$ (and $\{\ell, \ell'\}\in \calC^+$). It follows that $L'$ is a also a cover for $\calC$. To show that $L'$ is of minimum size assume for a contradiction that a smaller set $L''$ is also a cover for $\calC$.
Clearly, $L''$ covers $\calC \cap \calC_1$. Now take any set $\{\ell, \ell'\}$ in $\calC^+$, if neither $\ell$ nor $\ell'$ is part of $L''$, then we claim that $|L''| > k$. Indeed, we introduced $\{\ell, \ell'\}$ into $\calC_1$ only when more than $2k$ subsets in $\calC$ contain both $\ell$ and $\ell'$. If neither $\ell$ nor $\ell'$ are part of $L''$, then Lemma~\ref{lem:common3} implies that every other line can cover at most two of these subsets. In particular, the cardinality of $L''$ will be larger than $k$, contradicting with the fact that $L''$ is smaller than $L'$. A similar argumentation shows that any minimum-size cover of $\calC$ is also a minimum-size cover of $\calC_1$, which concludes the proof.
\end{proof}

%The set $\calC_1$ now contains subsets, where any two line segments are shared by at most $2k$ sets.

Now we further reduce our problem instance to a set~$\calC_2$ as follows. We count for each line how many subsets of~$\calC_1$ contain it. Then for each line~$\ell$ that has more than~$2k^2$ subsets containing it we replace all these subsets by subset~$\{\ell\}$.

\begin{lemma}\label{lem:common1}
A set $L'$, with $|L'| \leq k$ is a minimum-size cover for $\calC_1$ if and only if it is a minimum-size cover of~$\calC_2$.
\end{lemma}
\begin{proof}
As before, it suffices to consider the case in which $\calC_1\neq \calC_2$. Let $\calC_1^+ = \calC_2 \backslash \calC_1$ and $\calC_1^- = \calC_1 \backslash \calC_2$ denote the sets that were added and removed from $\calC_1$ to create $\calC_2$.
Since all sets in $\calC_1^-$ contain a line of one of the singleton sets of $\calC_1^+$, a set $L'$, with $|L'| \leq k$ that is a minimum size cover of $\calC_2$ is also a cover of $\calC_1$. To show that $L'$ is also a minimum-size cover for $\calC_1$, assume for a contradiction that there is a smaller cover $L''$ for $\calC_1$. The lines of $L''$ also cover $\calC_2 \backslash \calC_1^+$. Therefore, if $L''$ is not a cover of $\calC_2$, then there must be a subset $\{\ell\} \in \calC_1^+$ that is not covered by $L''$. Recall that $\{\ell\}$ was added to $\calC_2$ because there were more than $2k^2$ sets in $\calC_1$ that contain $\ell$. By construction of $\calC_1$, no line $\ell' \neq \ell$ can cover more than $2k$ of these sets (otherwise, the pair $\{\ell,\ell'\}$ would have been added to $\calC_1$). Thus, we conclude that $L''$ has more than $k$ segments, a contradiction.

In a similar way we can prove that a subset $L'$, with $|L'| \leq k$, that is a minimum-size cover for $\calC_1$ is also a minimum size cover for $\calC_2$, thus, concluding the proof.
\end{proof}

%With the previous three lemmas, we are now ready to prove the following on the kernel size for the problem.

\begin{lemma}
If $|\calC_2| > 2k^3$, then there is no cover of size at most $k$ for $\calC$.
\end{lemma}
\begin{proof}
Proof of this claim follows from Lemmas~\ref{lem:common2}~and~\ref{lem:common1} and the fact that each segment can only cover at most $2k^2$ sets of $\calC_2$.
\end{proof}

Now we look at the computational aspect of generating $\calC_2$. Both reduction steps from $\calC$ to $\calC_1$ and from $\calC_1$ to $\calC_2$ require counting subsets. Here we use the fact that the subsets are of size at most $4$ to show that this can be done in linear time using hash tables. Note that linear time here is in the size of $\calC$, the size of the arrangement, which may be quadratic with respect to the number of line segments.

\begin{lemma}
The set $\calC_2$ can be constructed in time $O(n\log n + C)$, where $n$ is the number of segments in $L$ and $C$ is the number of cells in the arrangement induced by $L$.
\end{lemma}
\begin{proof}
To reduce $\calC$ to $\calC_1$ we count for every pair of line segments $\ell, \ell'$ the number of subsets of $\calC$ that contain both. We do this by making a single pass over all subsets of $\calC$ and maintaining for each pair $\ell, \ell'$ how many subsets contain them thus far. To avoid having to initialize counts for all pairs $\ell, \ell'$, which may be more than linear in the size of $\calC$ we use a hash table and create a new count whenever we encounter a new pair of line segments. Each subset contains at most $4$ segments and at most $6$ pairs of segments, so we can process it in $O(1)$ time.
After counting we can go through all pairs of line segments with non-zero count and for each pair $\ell, \ell'$ that is contained in more than $2k$ subsets we remove the sets (which is easily done by storing which sets contain the pair) and add a new subset $\{\ell, \ell'\}$.
To compute $\calC_2$ from $\calC_1$ we can use a similar construction.
\end{proof}

%In summary we conclude the following.

\begin{theorem}\label{thm:kernel}
For the problem of finding a minimum-size cover for all rectangular cells in an arrangement of $n$ axis-parallel line segments $L$ we can either find a kernel of size $O(k^3)$ or conclude that no solution of at most size $k$ is possible. Moreover, the algorithm runs in $O(n\log n + C)$ time, where $C$ is the number of cells in the arrangement induced by $L$.
\end{theorem}

Since we now have a kernel of size $O(k^3)$ it follows that the problem is fixed parameter tractable~\cite{downey}.

\begin{corollary}\label{cor:fpt-alg}
Given a set $L$ of line segments the problem of finding a minimum size set $L' \subseteq L$ that covers the rectangular cells of the arrangement of $L$ is fixed parameter tractable with respect to the size $k$ of the optimal solution.
\end{corollary}

%%%%%%%%%%%%%%%%%%%%%%%%%%%%%%%%%%%%%%%%%%%%%%%%%%%%%%%
\section{Rectangular subdivisions}
\label{sec:rect-bsp}
%%%%%%%%%%%%%%%%%%%%%%%%%%%%%%%%%%%%%%%%%%%%%%%%%%%%%%%

Although the problem of finding a minimum set of covering segments is NP-hard, there are special cases where the problem can be solved in polynomial time. One such case is that the input line segments form a special type of rectangular subdivision. The rectangular subdivisions we consider are those defined by a KD-tree~\cite{bcko}. That is, a recursive subdivision of a rectangle using horizontal or vertical splitting segments (see Fig.~\ref{fig:rectilinear-bsp}). Note that the segments that have only an endpoint on the boundary of a cell are not considered part of the boundary of that cell. From now on we refer to such a subdivision simply as a \emph{rectangular subdivision}.

\begin{figure}
\centering
\includegraphics{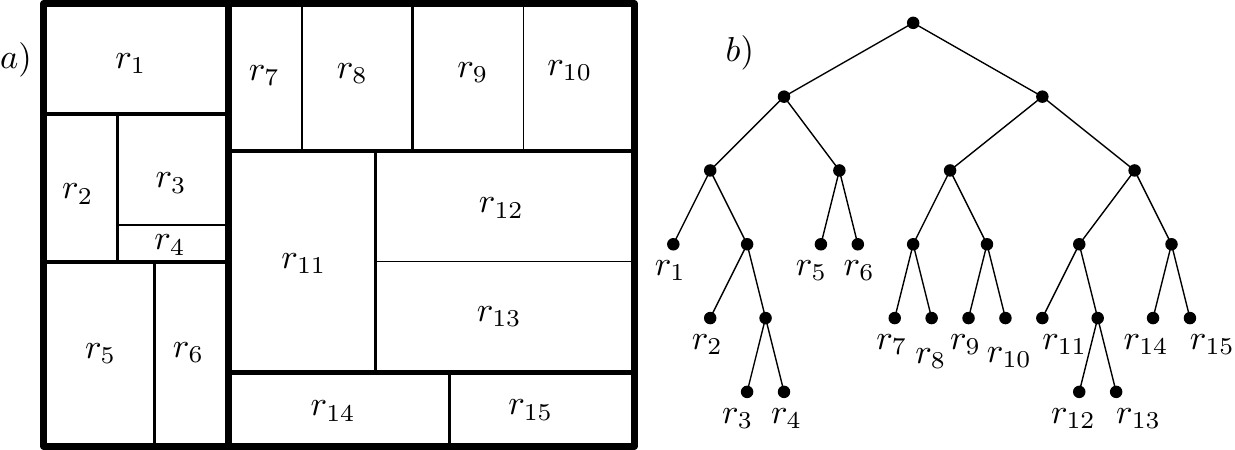}
\caption{a) A rectilinear binary space partition within a rectangle. b) A possible tree representing the partition.}
\label{fig:rectilinear-bsp}
\end{figure}

A rectangular subdivision provides a clear tree-structure which we can use to compute an optimal covering in a bottom up fashion. Each node in the tree is associated to some rectangle. For a leaf-node this rectangle is a cell of the arrangement, whereas for an interior node its associated rectangle $r$ is formed by the union of the rectangles associated to its children. We also associate an interior node of the tree with a horizontal or vertical segment that splits its associated rectangle into two rectangles associated to its children.

Although the above FPT approach works we show there is a much faster exact algorithm for rectangular subdivision. Without loss of generality we assume that the subdivision is given as a binary tree (the KD-tree structure)\footnote{If the structure is not given as a binary tree, but a more general subdivision structure such as a double-connected edge list, we can construct the tree in linear time.}. We show that an optimal covering can then be computed in linear time. Note that in this definition the outer face is covered if and only if at least one of the edges of the bounding rectangle is in the cover, and any other cell is covered if one of the segments defining its boundary is chosen as a covering segment.

\begin{theorem}
Given the tree structure of a rectangular subdivision, where each node stores the rectangle it represents and its splitting segment, we can compute in linear time a segment-cover of the cells with a minimum size.
\end{theorem}
\begin{proof}
Starting at the leaves, we compute an optimal covering in a bottom up fashion. For each node $v$ of the tree we compute the solution to sixteen different subproblems and store these solutions in the corresponding nodes. Let $R$ be the rectangle associated to a given node, and let $\{s_1, \ldots, s_4\}$ be the four segments that define its boundary. For any subset $S' \subseteq \{s_1, \ldots, s_4\}$, we consider the subproblem of finding the smallest covering of all the cells within $R$ that contains the segments of $S'$. For each such subproblem its cardinality as well as how it is constructed from solutions of its children (this second part is needed to reconstruct the optimal solution).

Clearly, if $v$ is a leaf the optimal cover is simply $S'$ (unless $S'=\emptyset$ in which case there is no solution). For an interior node we proceed as follows: let $v$ be an interior node of the tree and $R$ the rectangle stored at $v$. Without loss of generality assume that $R$ is split by a vertical line $\ell$. The children of $v$ are $v_{\myleft}$ and $v_{\myright}$ with corresponding rectangles $R_{\myleft}$ and $R_{\myright}$. We must compute a minimum-size cover for each possible choice of the top, left, bottom and right edges of $R$. Note that a fixed choice of boundary edges for $r$ already forces a choice of three edges for $R_{\myleft}$ and $R_{\myright}$. Only their shared edge $\ell$ is not fixed by the choice of boundary edges of $r$. However, we can simply try both options and see which results in the overall better cover of $R$. That is, we first assume $\ell$ is not part of the cover and retrieve our already computed solutions for $R_{\myleft}$ and $R_{\myright}$ for the current selection of boundary edges. Then we do the same when we do pick $\ell$ and choose the solution with the smallest number of edges. This results in an optimal solution since the only edges shared by $R_{\myleft}$ and $R_{\myright}$ are $\ell$ and the top and bottom edge of $R$. We consider all possible selections of these edges. For each selection the two subproblems of finding an optimal cover for the $R_{\myleft}$ and $R_{\myright}$ are independent, so we can reuse our previously computed solutions.

We now show that the algorithm indeed runs in linear time. Observe that only a constant number of subproblems are considered at each node of the tree. Moreover, each of these subproblems is solved in constant time by accessing the solution to a constant number of subproblems. Thus, overall we spent a constant amount of time per node of the tree, giving the desired bound.
\end{proof}

\section{Conclusions and open problems}

Our results show that covering cells in an arrangement, similar to the original set-cover problem, may be NP-hard or polynomial-time solvable depending on various restrictions. It may be interesting to investigate further variants of the problem to see which restriction makes the problem polynomial-time solvable---this may be due to the lack of intersections between line segments or due to the tree-structure of the subdivision. It would also be interesting to see if good approximations are possible for the more general case or even the case with line segments of arbitrary orientations.

\small{

}

\end{document}